\documentclass[12pt]{amsart}

\usepackage[usenames]{color}

\usepackage{amsmath}
\usepackage{amsthm}
\usepackage{amsfonts}
\usepackage{amssymb}
\usepackage{latexsym}
\usepackage{color}
\usepackage{threeparttable}
\usepackage{seqsplit}

\usepackage{hyperref}
\usepackage[top=1in, bottom=0.89in, left=1in, right=1in]{geometry}

\newtheorem{problem}{Problem}

\newtheorem{proposition}{Proposition}

\newcommand{\Z}{\mathbb{Z}}

\newcommand{\Q}{\mathbb{Q}}

\newcommand{\F}{\mathbb{F}}

\newcommand{\R}{\mathbb{R}}

\begin{document}

\author[]{Vladimir Shpilrain}
\address{Department of Mathematics, The City College of New York, New York,
NY 10031} \email{shpilrain@yahoo.com}

\author[]{Bianca Sosnovski}
\address{Department of Mathematics and Computer Science, Queensborough Community College, City University of New York, Bayside, NY, 11364}
\email{bsosnovski@qcc.cuny.edu}

\title{Cayley hashing with cookies}

\begin{abstract}
Cayley hash functions are based on a simple idea of using a pair of
semigroup elements,  $A$ and  $B$, to hash the 0 and 1 bit,
respectively, and  then to hash an arbitrary bit string in the
natural way, by using multiplication of elements in the semigroup. The main advantage of Cayley hash functions compared to, say, hash functions in the SHA family is that when an already hashed document is amended, one does not have to hash the whole amended document all over again, but rather hash just the amended part and then multiply the result by the hash of the original document. Some authors argued that this may be a security hazard, specifically that this property may facilitate finding a second preimage by splitting a long bit string into shorter pieces. In this paper, we offer a way to get rid of this alleged disadvantage and keep the advantages at the same time. We call this method  ``Cayley hashing with cookies" using terminology borrowed from the theory of random walks in a random environment. For the platform semigroup, we use $2\times 2$ matrices over $\F_p$.

\end{abstract}

\maketitle

\section{Introduction}

Hash functions are easy-to-compute compression functions  that take
a variable-length input and convert it to a fixed-length output.
Hash functions are used as compact representations, or digital
fingerprints, of data and to provide message integrity. Basic security
requirements are well known:

\begin{enumerate}

\item {\it Collision resistance}: it should
be computationally infeasible to find two different inputs that hash
to the same output.

\item {\it Preimage resistance} (sometimes called {\it non-invertibility}): it should be computationally
infeasible to find an input which hashes to a specified output.

\item {\it Second preimage resistance}: it
should be computationally infeasible to find a second input that
hashes to the same output as a specified input.

\end{enumerate}

A challenging problem is to determine mathematical properties of a
hash function that would ensure (or at least, make it likely) that
the requirements  above are met.

A direction that has been gaining momentum lately is using a pair of
elements, $A$ and $B$, of a semigroup $S$, to hash the ``0" and the ``1" bit, respectively. Then a bit string is hashed to a product of elements in the natural way.
For example, the bit string 1001011 will be hashed to  the element $BAABABB$.

Since hashing a random bit string this way represents a random walk on the Cayley
graph of the subsemigroup of $S$ generated by the elements $A$ and $B$, hash functions of this kind are often called {\it Cayley hash functions}. Note that the absence of short collisions for a Cayley hash function is equivalent to the corresponding Cayley graph having a large {\it girth}. The latter is defined as the length of the shortest circuit.

Cayley hash functions have a homomorphic property $H(XY)=H(X)H(Y)$ and the
associativity property $H(XYZ)=H(XY)H(Z) = H(X)H(YZ)$ for any bit
strings $X, Y, Z$. (Here $XY$ means concatenation of the bit strings $X$
and $Y$.) This property is useful not only because it allows for parallel computations  when hashing a long bit string. A more important feature is: when an already hashed document is amended, one does not have to hash the whole amended document all over again, but rather hash just the amended part and then multiply the result by the hash of the original document. On the flip side, this property may to some extent facilitate finding a second preimage by splitting a long bit string into shorter pieces.

Another useful property of a Cayley hash function is that, unlike with a SHA hash function, you do not have to know the length of a bit string to be hashed up front; you can hash ``as you go".

Needless to say, while the high-level idea of Cayley hashing is definitely appealing, the choice of the platform semigroup $S$ and two elements $A, B \in S$ is crucial for security and efficiency. There have been many proposals based on matrix semigroups in $GL_2(\F)$ for various fields $\F$, in particular for $\F=\F_p$. This is because Cayley graphs of 2-generator semigroups in $GL_2(\F_p)$ tend to have a large girth as was shown by several authors, see e.g. \cite{BG}, \cite{BSV}, \cite{H}, \cite{Larsen}.

Cayley graphs of (semi)groups in $GL_n(\F_p)$ with $n >2$ have been considered, too (see \cite{Jo}, \cite{AB}, \cite{Battarbee}), but we will focus here on $n =2$ for the reasons outlined in our Section \ref{growth}; one obvious reason is a smaller size of the hash. For example, if $p$ is a 256-bit prime, then any matrix from $GL_2(\F_p)$ has size of up to 1024 bits, which is common for standard hash functions these days, e.g. for the SHA family.

The novel contribution of the present paper is introducing what we call ``Cayley hashing with cookies", the terminology borrowed from the theory of random walks in a random environment, see e.g. \cite{BS08},  \cite{Solomon75}. We argue that this enhancement does not affect the collision resistance property, and at the same time makes the hash function more preimage resistant. The homomorphic property is ``almost preserved", i.e., is preserved upon minor padding. The corresponding hashing protocol is described in Section \ref{cookies}, and the girth of the relevant Cayley graph is discussed in Section \ref{growth}. Efficiency is discussed in Section \ref{efficiency}.

\section{Background}\label{previous}

The first proposal of a Cayley hash function was due to Z\'emor \cite{Zemor}.
The matrices used, considered over $\F_p$, were \begin{displaymath}
A = \left(
 \begin{array}{cc} 1 & 1 \\ 0 & 1 \end{array} \right) , \hskip 1cm B = \left(
 \begin{array}{cc} 1 & 0 \\ 1 & 1 \end{array} \right).
\end{displaymath}

This proposal was successfully attacked in \cite{TZ_attack}. Specifically, it was shown that this hash function is not preimage resistant.

The most cited proposal is what has become known as the Tillich-Z\'emor hash function \cite{TZ}. Their matrices were \begin{displaymath}
A = \left(
 \begin{array}{cc} \alpha & 1 \\ 1 & 0 \end{array} \right) , \hskip 1cm B = \left(
 \begin{array}{cc} \alpha & \alpha+1 \\ 1 & 1 \end{array} \right).
\end{displaymath}

\noindent  These matrices are considered over a field defined as
$R=\F_2[x]/(p(x))$, where  $\F_2[x]$ is the ring of polynomials over
$\F_2$,  $(p(x))$ is the ideal of $\F_2[x]$
generated by an irreducible polynomial $p(x)$ of degree $n$
(typically,  $n$ is a prime, $127 \le n \le 170$), and  $\alpha$ is a root of $p(x)$.

The reason for selecting such a ``fancy" field probably was to specifically avoid the attack in \cite{TZ_attack}.

Similar later proposals include \cite{Abdukhalikov}, \cite{Joju}, \cite{Sosnovski},  \cite{Tomkins}. Several attacks (some of them targeted at finding collisions, some targeted at finding a preimage) were suggested over the years \cite{GIMS}, \cite{MT}, \cite{PQ}, \cite{Petit_Rubik}.

Another idea for avoiding short collisions is to use a pair of $2 \times
2$ matrices, $A$ and $B$, over $\Z$  that generate a free semigroup in $GL_2(\Z)$,
and then reduce the entries modulo a large prime $p$ to get matrices
over $\F_p$. Since there cannot be an equality of two different
products of copies of $A$ and $B$ unless at least one of
the entries in at least one of the products is $\ge p$, this gives a
lower bound on the minimum length of bit strings where a collision
may occur.

\section{Cayley hashing with cookies}\label{cookies}

Inspired by the theory of random walks with cookies (see e.g. \cite{BS08},  \cite{Solomon75}), here we introduce a Cayley hash function with cookies.

Let $A$, $B$, and $C$ be $2 \times 2$ matrices. Let $u$ be a bit string of an arbitrary length. Then, to hash $u$, going left to right:

\bigskip

\noindent {\bf 1.} If the current bit is 0, then it is hashed to the matrix $A$.
If the current bit is 1, then it is hashed to the matrix $B$. 

\medskip

\noindent {\bf 2.} If there are three ``1" bits in a row (a ``cookie"), then all following ``1" bits will be hashed to the matrix $C$, until there are three ``0" bits in a row, in which case hashing the ``1" bit is switched back to the matrix $B$.
For example, the bit string 10011110001 will be hashed to the matrix $BAABBBCAAAB$.

\bigskip

\noindent The recommended particular matrices are: $A = \left(
 \begin{array}{cc} 1 & 2 \\ 0 & 1 \end{array} \right) , \hskip .2cm B = \left(
 \begin{array}{cc} 1 & 0 \\ 2 & 1 \end{array} \right), \hskip .2cm C = \left(
 \begin{array}{cc} 2 & 1 \\ 1 & 1 \end{array} \right)$.

\begin{proposition}\label{free} The semigroup generated by the matrices $A$, $B$, and $C$ over $\Z$ is free.

\end{proposition}

\begin{proof}
Denote $X = \left(
 \begin{array}{cc} 1 & 1 \\ 0 & 1 \end{array} \right) , \hskip .2cm Y = \left(
 \begin{array}{cc} 1 & 0 \\ 1 & 1 \end{array} \right).$

It is well known that  $X$ and $Y$ generate  a free semigroup.
Then, $A = X^2, B=Y^2, C=XY$. Note that none of the three words  $X^2, Y^2, XY$ is a prefix of any other. In that case, it is known (and easy to see) that there are no semigroup relations between such words.

\end{proof}

Thus, if our matrices $A, B, C$ are considered over $\F_p$, there cannot be any collisions in the corresponding hash function unless a bit string that is hashed is long enough for at least one of the entries in a product of matrices to become larger than $p$. This is why it is important to determine the growth of the largest entry in a product of $n$ matrices, as a function of $n$. This is what our Section \ref{growth} is about.

\subsection{Padding}\label{Padding} To preserve the useful homomorphic property $H(XY)=H(X)H(Y)$ of the hash function $H$ (see the Introduction), in our situation one has to do a minor padding of any bit string to be hashed. Specifically, three zeros would have to be added at the end of each bit string to be hashed. Because of the rules at Step 2 of the hashing protocol, this will reset hashing elements to the original pair $(A, B)$ of matrices.

\section{Efficiency}\label{efficiency}

With the particular choice of matrices $A, B, C$ as in the previous Section \ref{cookies}, computation of the hash is very efficient. Indeed, computing the hash $H(u)$ of a given bit string $u$ of length $n$ takes $(n-1)$ matrix multiplications where each time one multiplies by one of the matrices $A$, $B$, or $C$.

Now we note that in any of the matrices $A, B, C$ one of the entries is 2 and other entries are 0 or 1. When multiplying by a matrix like that, we do not actually have to do any multiplications of numbers since multiplying a number $x$ by 2 is the same as adding $x$ to itself.

With this in mind, we see that multiplying by any of the matrices $A$, $B$, or $C$  requires 5 additions of numbers. Therefore, computing $H(u)$ requires no multiplications and  $5(n-1)$ additions in $\F_p$.

\section{Growth}\label{growth}

In this section, we address the following general problems about the growth of the maximal entry in a product of $n$ matrices. The motivation is: the slower the growth, the longer are minimal collisions in the corresponding hash function, as explained at the end of Section \ref{previous}, as well as at the end of Section \ref{cookies}.

In what follows, $A, B, C$ are matrices over $\Z$.

\begin{problem}\label{maximum} What is the maximal possible entry of a matrix $w(A, B, C)$, as a function of the word length $n=|w|$, over all words $w$ of length $n$?
\end{problem}

\begin{problem}\label{generic}
What is the maximal entry of a matrix $w(A, B, C)$, where $w$ is a generic word of length $n$?
\end{problem}

The problem with Problem \ref{generic} in reference to our situation in Section \ref{cookies} is that it is not easy to formalize ``generic" since the probability of matrices $B$ or $C$ appearing in any particular place in a product of $n$ matrices is difficult to estimate. However, Problem \ref{generic} can be studied experimentally.

The growth of entries in 3-generator semigroups of matrices has not been studied before, to the best of our knowledge. By comparison, the 2-generator case has been studied fairly well.
Below we summarize what is known about the growth of entries in 2-generator semigroups of matrices over $\Z$, i.e., in matrices of the form  $w(A, B)$ for various popular instances of $A$ and $B$.

Denote $A(k) = \left(
 \begin{array}{cc} 1 & k \\ 0 & 1 \end{array} \right) , \hskip .2cm B(m) = \left(
 \begin{array}{cc} 1 & 0 \\ m & 1 \end{array} \right)$.

\medskip

\noindent {\bf 1.} In \cite{BSV}, it was proved that the maximum growth in products of $n$ matrices of the form $w(A(k), B(k))$ for integers $k \ge 1$ is achieved by the words $w = (A(k)B(k))^{\frac{n}{2}}$ (assuming that $n$ is even).

\medskip

\noindent {\bf 2.} As one would expect, growth of the entries in matrices $w(A(1), B(1))$ is the slowest among all $w(A(k), B(k))$ for integers $k \ge 1$. The largest entries in the corresponding matrices $(A(1)B(1))^{\frac{n}{2}}$ are $O((\frac{3}{2}+\frac{\sqrt{5}}{2})^n)$. Note that $\frac{3}{2}+\frac{\sqrt{5}}{2} \approx 1.618$.

\medskip

\noindent {\bf 3.} For matrices $w(A(2), B(2))$, the largest entries are in the matrices $(A(2)B(2))^{\frac{n}{2}}$, and their magnitude is $O((1+\sqrt{2})^n)$, see \cite{BSV}. This implies that the girth of the corresponding Cayley graph (over $\F_p$) is $O(\log_{1+\sqrt{2}} n)$. Note that $1+\sqrt{2} \approx 2.41$. Incidentally, this is the best known lower bound for the girth in this particular case. In fact, in \cite{BSV} an exact expression for the largest entries was obtained: $(\frac{1}{2}+\frac{1}{\sqrt{8}})(1+\sqrt{2})^n +
(\frac{1}{2}-\frac{1}{\sqrt{8}})(1-\sqrt{2})^n$.

This implies, in particular, that if $p$ is on the order of $2^{256}$,  then there
are no collisions of the form  $u(A(2), B(2)) = v(A(2), B(2))$ if both the
words $u$ and  $v$ are of length less than $203 \approx \log_{2.41}2^{256} = 256 \log_{2.41}2$.

We also note that up to date, there have been no successful attacks reported against the hash function based on the matrices $A(2)$ and $B(2)$.
\medskip

\noindent {\bf 4.} The pair of matrices $A=A(2)$ and $B=B(-2)$ give the minimum growth rate of the largest entry in $w(A(k), B(m))$ among all $|k|, |m| \ge 2$. According to computer experiments, the largest (by the absolute value) entries occur in $(ABBA)^{\frac{n}{4}}$, and these entries are $O((\sqrt{2+\sqrt{3}})^n)$. Note that $\sqrt{2+\sqrt{3}} \approx 1.93$.

This implies, in particular, that if $p$ is on the order of $2^{256}$,  then there
are no collisions of the form  $u(A(2), B(-2)) = v(A(2), B(-2))$ if both the
words $u$ and  $v$ are of length less than $269 \approx  256 \log_{1.93}2$.
\medskip

\noindent {\bf 5.} Generically, i.e., in a random product of length $n$ of the matrices $A(2)$ and $B(2)$ (where each factor is $A(2)$ or $B(2)$ with probability $\frac{1}{2}$), the largest entry grows approximately as $(1.9)^n$. This was determined experimentally, by averaging over 1000 products of 1000 matrices.
\medskip

\noindent {\bf 6.} Generically, the largest entry in a random product of length $n$ of the matrices $A(2)$ and $B(-2)$ grows approximately as $(1.68)^n$. Again, this was determined experimentally, by averaging over 1000 products of 1000 matrices.

\subsection{Growth in $w(A, B, C)$}\label{growth3}
Now we get to the growth questions (Problems \ref{maximum} and \ref{generic})
that are relevant to our particular Cayley hash function from Section \ref{cookies}.
Recall that in our situation $A = X^2, ~B=Y^2, ~C=XY$, where $X = \left(
 \begin{array}{cc} 1 & 1 \\ 0 & 1 \end{array} \right) , \hskip .2cm Y = \left(
 \begin{array}{cc} 1 & 0 \\ 1 & 1 \end{array} \right).$

\begin{proposition}\label{fastest} The bit string $11111\ldots  $ yields hash matrices with the fastest growing entries, among hash matrices of all bit strings of the same length.

\end{proposition}

\begin{proof}
The hash matrix of such a bit string of length $n$ is $B^3C^{n-3}$. Neglecting the $B^3$ factor, what we have here is powers of $C=XY$, so powers of $C$ are alternating products of the matrices $X$ and $Y$. The latter are known to give the fastest growth among all $w(X, Y)$, see item (1) in Section \ref{growth}.

Throwing in some 0 bits in this bit string will result in throwing in some $X^2$ matrices in the matrix product, and then we will have subfactors like $XYX^2$ and/or $X^2XY=X^3Y$. In either case, $X$ and $Y$ will no longer be alternating in a product, so the growth of the entries in a product matrix will be slower.

Alternatively, if, after throwing in some 0 bits we get three zeros in a row, hashing the 1 bit will be switched to the matrix $B=Y^2$, and then we will have subfactors that are products of matrices $X^2$ and $Y^2$, so again $X$ and $Y$ will not be alternating in such a product, so again the growth of the entries in a product matrix will be slower than it is in the matrices that hash a sequence of 1 bits only.
\end{proof}

Thus, the maximum growth is that of the entries of the matrices $C^n=(XY)^n$, and this is known to be on the order of
$O((\frac{3}{2}+\frac{\sqrt{5}}{2})^{2n}) = O((\frac{7}{2}+\frac{3}{2}\sqrt{5})^{n})$. Note that $\frac{7}{2}+\frac{3\sqrt{5}}{2} \approx 2.618$.

This implies that there are no collisions $H(u)=H(v)$ in our hash function $H$ if both
bit strings $u$ and  $v$ are of length less than $\log_{2.618} p$. In particular,
if $p$ is on the order of $2^{256}$,  then there are no collisions if both
bit strings $u$ and  $v$ are of length less than $184 = 256 \log_{2.618}2$.

\section{Collision and preimage resistance} \label{security}

Collision resistance claims for Cayley hash functions are typically based on satisfactory lower bounds for the girth of the relevant Cayley graph. Our lower bound is logarithmic in $p$, see Section \ref{growth3}, which is consistent with other proposals of Cayley hash functions (see e.g. \cite{Jo}, \cite{AB},   \cite{Joju}, \cite{Battarbee},  \cite{Sosnovski},  \cite{Tomkins}) that use matrices over the field $\F_p$ or its extensions. Of course, the base of the logarithm matters, too, which is why more specific lower bounds on the girth (as in our Section \ref{growth}) are important.

Our method of estimating the girth is described at the end of Section \ref{cookies}. It gives very good results in some cases; in particular, for the girth of the Cayley graph corresponding to the pair of matrices $(A(2), B(2))$ (see Section \ref{growth}) our lower bound is tighter than lower bounds obtained by other authors (\cite{BG}, \cite{H}, \cite{Larsen}). However, this has a flip side: unless some of the entries in a matrix $w(A(2), B(2))$ are larger than $p$, this $w(A(2), B(2))$ is an element of the free semigroup in $SL_2(\Z)$ generated by $A(2)$ and $B(2)$. In that case, there is an efficient algorithm \cite{Geller} that recovers the word $w$, i.e., in the context of the corresponding Cayley hash function it recovers a preimage of the hash. The reason why this algorithm is efficient is that, given a matrix $W=w(A(2), B(2))$, exactly one of the matrices $WA^{-1}$ and  $WB^{-1}$ has a smaller sum of the absolute values of the entries than the matrix $W$ does.

With our hash function, this algorithm typically will not be feasible even in case of relatively short bit strings (of a couple of hundred bits). This is because in our situation it is not true that, given a matrix $W=w(A, B, C)$, exactly one of the matrices $WA^{-1}$, $WB^{-1}$, and  $WC^{-1}$ has a smaller sum of the absolute values of the entries than the matrix $W$ does. Therefore, at least at some steps the attacker would have to explore more than one option, so the number of steps can be exponential in the number of bits in a bit string that the attacker wants to recover.

With a brute force attack (trying out, one at a time, all $2^n$ bit strings of length $n$ until a preimage is found), the number of trials is on the order of $2^n$. 
Thus, for preimage security, the length of a bit string to be hashed has to be at least $t$, where $t$ is the security parameter. Currently, it is recommended that $t \ge 256$.

We also mention that to date, there were no successful attacks reported against the Cayley hash function, call it $H_1$, based on the two matrices $A = \left(
 \begin{array}{cc} 1 & 2 \\ 0 & 1 \end{array} \right) , \hskip .2cm B = \left(
 \begin{array}{cc} 1 & 0 \\ 2 & 1 \end{array} \right)$, see \cite{BSV}. The Cayley hash function in the present paper, call it $H_2$, based on the matrices $A, B$, and $C = \left(\begin{array}{cc} 2 & 1 \\ 1 & 1 \end{array} \right)$, is at least as preimage resistant as $H_1$ is, in the following sense: if there is an algorithm (deterministic or not) for recovering preimage of $H_2(u)$ for any bit string $u$ of length $n$, then the same algorithm will recover preimage of $H_1(u)$ for any bit string $u$ of length $n$. This is because if a bit string $u$ does not include a substring of three ``1" bits in a row, then $H_1(u)$ is just the same as $H_2(u)$.

\section{Suggested parameters}\label{parameters}

\noindent For $p$ in $\F_p$, we suggest a 256-bit prime.

\noindent For matrices that hash individual bits, we suggest $A = \left(
 \begin{array}{cc} 1 & 2 \\ 0 & 1 \end{array} \right) , \hskip .2cm B = \left(
 \begin{array}{cc} 1 & 0 \\ 2 & 1 \end{array} \right), \hskip .2cm C = \left(
 \begin{array}{cc} 2 & 1 \\ 1 & 1 \end{array} \right)$.

\section{NIST statistical test suite results}
\label{NIST}

A hash function should generate outputs as random as possible. We applied the NIST Statistical Test Suite \cite{NIST} to evaluate the randomness of the outputs in binary form from the proposed hash function.

The NIST Statistical Test Suite is a package that includes 15 types of tests, each with a suitable metric needed to investigate the degree of randomness for binary sequences produced by cryptographic random generators.  In these tests, a set of statistical tests for randomness are used for detecting deviations of a binary sequence from randomness.

Even though no statistical test can certify if a (pseudo)random generator is suitable for usage in a specific cryptographic application,  the NIST tests may be useful as a first step in that direction. 

The NIST statistical tests are formulated to test the null hypothesis $H_0$ that the sequence being tested is random. Thus, the alternative hypothesis $H_A$ is that the sequence being tested is non-random. For each NIST test and sequence tested, a test statistic value is calculated from the sample of bits.

The $P$-value (or probability value) is the probability of getting a sample statistic with the test value or a more extreme sample statistic in the direction of the alternative hypothesis $H_A$ under the assumption that the null hypothesis $H_0$ is true. It is a measure of strength of the evidence against the null hypothesis (randomness). Specifically, if the $P$-value is $\geq \alpha$ (NIST suite has $ \alpha$ set to $0.01=1\%$), the conclusion is that the sequence is random, otherwise, it is non-random.


The following are deviations from randomness that each test in the NIST Suite detects  in binary sequences:

\begin{itemize}
\item \emph{Frequency test}  -  Too many zeroes or ones.
\item \emph{Block frequency test}  -  Too many zeros or ones within a block
\item \emph{Runs test } -  Large (small) total number of runs indicates that the oscillation in the bit string is too fast (too slow).
\item \emph{Longest runs of ones test} - Deviation of the distribution of long runs of ones.
\item \emph{Rank test }-  Deviation of the rank distribution from a corresponding random sequence, due to periodicity.
\item \emph{Discrete Fourier Transform (spectral) test} -  Periodic features in the bit stream.
\item \emph{Non-overlapping template matchings test} - Too many occurrences of non-periodic templates.
\item \emph{Overlapping template matchings test }- Too many occurrences of $m$-bit runs of ones.
\item \emph{Universal statistical test} -  Compressibility (regularity).
\item \emph{Linear complexity test} -  Deviation from the distribution of the linear complexity for finite length (sub)strings.
\item \emph{Serial test} -  Non-uniform distribution of $m$-length words. Similar to the approximate entropy test.
\item \emph{Approximate entropy test}  -  Non-uniform distribution of $m$-length words. Small values of ApEn(m) imply strong regularity.
\item \emph{Cumulative sums test }- Too many zeroes or ones at the beginning of the sequence.
\item \emph{Random excursions test}  -  Deviation from the distribution of the number of visits of a random walk to a certain state.
\item \emph{Random excursion variant test }-  Deviation from the distribution of the total number of visits (across many random walks) to a certain state.
\end{itemize}

For each test, there is a recommended minimum size for the binary streams being tested \cite{RSN}.  If one wishes to apply all the tests in the suite, a minimum of $10^{6}$ in length is recommended for the binary strings tested. Also, the $P$-values processed by the NIST tests use approximation, so the more sequences are tested the more accurate results will be obtained. 

We used SageMath \cite{Sage} to generate the hash values in binary form. Random primes of order $2^{256}$ and $2^{512}$ were generated for the modulos of the hash function and also random binary strings as inputs of length $10^{6}$ bits. These inputs were padded with 000 to reset the matrices (see our Section \ref{Padding}), and finally, the corresponding matrix products were calculated.  The hash values consist of the concatenated matrix entries in binary form with lengths of 1024 and 2048 bits, respectively.
We have analyzed data for the modulus $p$ of the order $2^{256}$ and $2^{512}$. 

Tables \ref{table-nist1}  and \ref{table-nist2} present the statistical properties of the hash values as reported by the NIST test suite, obtained after processing 100 binary sequences of length $10^6$.

\begin{threeparttable}
\centering \caption{NIST Suite Results for a Prime of Order $2^{256}$}
\label{table-nist1}
  \begin{tabular}{c| l | c | c  }
  \hline \hline
Number & Statistical test & $P$-value   & Pass rate \\ \hline
\hline 1 & Frequency     & 0.678686  & 98/100  \\ \hline
2 & Block frequency  &  0.003447 &   100/100            \\ \hline
3 & Cumulative sums 1  &  0.224821 &    98/100    \\ \hline
4 & Cumulative sums 2  & 0.719747  &    97/100    \\ \hline
5 & Runs   &    0.021999  &  100/100         \\ \hline
6 & Longest runs of ones   &   0.289667 &  99/100   \\ \hline
7 & Rank  &   0.935716  &    99/100    \\ \hline
8 & FFT  &   0.045675   &    98/100   \\ \hline
9..156 & Non-overlapping templates  &  0.471367\tnote{1 }  {\small  \ (mean)}    &   98.92/100\tnote{2 } {\small  \ (mean)}    \\ \hline
157 & Overlapping template   &    0.304126  &    100/100   \\ \hline
158 & Universal  &  0.657933    &     99/100     \\ \hline
159 & Approximate entropy  &  0.224821  &   98/100    \\ \hline
160..167 & Random excursions  &   0.533178\tnote{3 } {\small (mean)}   &   50.75/51\tnote{4 } {\small (mean)}   \\ \hline
168..185 & Random excursions variant &  0.344685\tnote{5 } {\small (mean)}    &  50.78/51\tnote{6 } {\small (mean)}    \\ \hline
186 & Serial 1 & 0.514124  &  98/100     \\ \hline
187 & Serial 2 &  0.401199   &    99/100   \\ \hline
188 &Linear complexity &   0.249284   & 99/100     \\ \hline
\end{tabular}
\begin{tablenotes}
\item[1] 148 tests with a minimum of 0.006196 and a maximum of 0.99425.
\item[2] 148 tests with a minimum of 96 and a maximum of 100.
\item[3] 8 tests with a minimum of 0.032923 and a maximum of 0.964295.
\item[4] 8 tests with a minimum of 50 and a maximum of 51.
\item[5] 18 tests with a minimum of 0.048716 and a maximum of 0.719747.
\item[6] 18 tests with a minimum of 50 and a maximum of 51.
\end{tablenotes}
\end{threeparttable}

\begin{threeparttable}
\centering \caption{NIST Suite Results for a Prime of Order $2^{512}$}
\label{table-nist2}
  \begin{tabular}{c| l | c | c  }
  \hline \hline
Number & Statistical test & $P$-value   & Pass rate \\ \hline
\hline 1 & Frequency     & 0.071177  & 98/100  \\ \hline
2 & Block frequency  &  0.202268 &   99/100            \\ \hline
3 & Cumulative sums 1  &  0.304126 &    97/100    \\ \hline
4 & Cumulative sums 2  &  0.224821  &    97/100    \\ \hline
5 & Runs   &    0.759756  &  98/100         \\ \hline
6 & Longest runs of ones   &   0.366918 &  100/100   \\ \hline
7 & Rank  &   0.090936  &    100/100    \\ \hline
8 & FFT  &   0.798139   &    99/100   \\ \hline
9..156 & Non-overlapping templates  &  0.518710\tnote{1 }  {\small  \ (mean)}    &   99.09/100\tnote{2 } {\small  \ (mean)}    \\ \hline
157 & Overlapping template   &    0.637119  &    96/100   \\ \hline
158 & Universal  &  0.23681    &     100/100     \\ \hline
159 & Approximate entropy  &  0.062821  &   99/100    \\ \hline
160..167 & Random excursions  &   0.680366\tnote{3 } {\small (mean)}   &   61.38/62\tnote{4 } {\small (mean)}   \\ \hline
168..185 & Random excursions variant &  0.394883\tnote{5 } {\small (mean)}    &  61.5/62\tnote{6 } {\small (mean)}    \\ \hline
186 & Serial 1 & 0.334538  &  97/100     \\ \hline
187 & Serial 2 &  0.678686   &    100/100   \\ \hline
188 &Linear complexity &   0.249284   & 96/100     \\ \hline
\end{tabular}
\begin{tablenotes}
\item[1] 148 tests with a minimum of 0.004301 and a maximum of 0.996335.
\item[2] 148 tests with a minimum of 96 and a maximum of 100.
\item[3] 8 tests with a minimum of 0.437274 and a maximum of 0.862344.
\item[4] 8 tests with a minimum of 60 and a maximum of 62.
\item[5] 18 tests with a minimum of 0.039244 and a maximum of 0.985035.
\item[6] 18 tests with a minimum of 61 and a maximum of 62.
\end{tablenotes}
\end{threeparttable}

\

\

NIST suggests considering data to be random if and only if the sequence/sequences pass both the uniformity test of $P$-values and the test of the proportion of passing sequences.

According to the NIST documentation, a pass rate of 96\% is acceptable.
This corresponds to a minimum pass rate for each statistical test, except the random excursion (variant) test, which is approximately 96 for a sample size of 100 binary sequences. The minimum pass rate for the random excursion (variant) test is approximately 48 for a sample size of 51 binary sequences with $p$ of order $2^{256}$, and approximately 59 for a sample size of 62 binary sequences with $p$ of order $2^{512}$.

The following primes $p$ were used to generate hash values tested with the NIST
Statistical Test Suite.

\vskip -0.2cm

\begin{itemize}

\item Prime of order $2^{256}$:

Decimal form:
  \seqsplit{11213019353385680997044300082282941457293378
0556534369189742044710202716867171}

%
\medskip

\item Prime of order $2^{512}$:

Decimal form:
\seqsplit{125967099140123813315752220780255508336665456536865562994
12073058759112539196792509169699422775197821869177859263195
184957153059906758380302238329723774073}


\end{itemize}

\section{Conclusions}

$\bullet$ We have proposed a Cayley hash function $H$ that employs random walks (with cookies) on the Cayley graph of a 3-generator (instead of the usual 2-generator) semigroup of $2\times 2$ matrices over $\F_p$.
\medskip

$\bullet$ If $p$ is a 256-bit prime, then the size of $H(u)$ for any bit string $u$ is 1024 bits. If the bit string $u$ has $n$ bits, then computing $H(u)$ (with the recommended choice of matrices $A, B, C$) requires no multiplications and  $5(n-1)$ additions in $\F_p$.
\medskip

$\bullet$ There are provably no collisions $H(u)=H(v)$ in our hash function $H$ if both
bit strings $u$ and  $v$ are of length less than $\log_{_{\frac{7}{2}+\frac{3\sqrt{5}}{2}}} p \approx \log_{2.618} p$. In particular,
if $p$ is a 256-bit prime,  then there are no collisions if both
bit strings $u$ and  $v$ are of length less than $184 = 256 \log_{2.618}2$. If $p$ is a 512-bit prime,  then there are no collisions if both bit strings are of length less than 368, etc. We note that $\log_{_{\frac{7}{2}+\frac{3\sqrt{5}}{2}}} p$ is just a provable  lower bound for the girth of the relevant Cayley graph; the actual girth might be much larger.
\medskip

$\bullet$  Our hash function has successfully passed the applicable
pseudorandomness tests in the NIST Statistical Test Suite.


\end{document}